\documentclass[preprint,aps,pra,onecolumn,superscriptaddress,floatfix,nofootinbib,showpacs,longbibliography]{revtex4-1}
\usepackage[utf8]{inputenc}
\usepackage[T1]{fontenc}     
\usepackage[british]{babel} 
\usepackage[sc,osf]{mathpazo}
\usepackage{libertineRoman}  
\usepackage[colorlinks=true, citecolor=blue, urlcolor=blue]{hyperref}  
\usepackage{enumerate}
\usepackage{graphicx}
\usepackage[babel]{microtype}  
\usepackage{amsmath,amssymb,amsthm,bm,amsfonts,mathrsfs,bbm} 
\usepackage{physics}

\usepackage{xspace}  
\usepackage{diagbox}
\usepackage{xcolor}
\usepackage{multirow}
\usepackage{array}
\usepackage{bigstrut}
\usepackage{braket}
\usepackage{color}
\usepackage{natbib}
\usepackage{multirow}
\usepackage{mathtools}
\usepackage{float}
\usepackage{xcolor,colortbl}
\usepackage{physics}
\usepackage{amsmath}
\usepackage{color}
\usepackage[justification=justified, format=plain]{subcaption}
\usepackage[justification=raggedright]{caption}

\newcommand{\be}{\begin{equation}}
\newcommand{\ee}{\end{equation}}
\newcommand{\ba}{\begin{eqnarray}}
\newcommand{\ea}{\end{eqnarray}}

\newtheorem{theorem}{Theorem}
\newtheorem{corollary}{Corollary}

\newtheorem{lemma}{Lemma}

\def\>{\rangle}
\def\<{\langle}

\begin{document}
	
\title{Measurement-Device-Independent Schmidt Number Certification of All Entangled States}	

\author{Saheli Mukherjee}
\email{mukherjeesaheli95@gmail.com}
\affiliation{S. N. Bose National Centre for Basic Sciences, Block JD, Sector III, Salt Lake, Kolkata 700 106, India}

\author{Bivas Mallick}
\email{bivasqic@gmail.com}
\affiliation{S. N. Bose National Centre for Basic Sciences, Block JD, Sector III, Salt Lake, Kolkata 700 106, India}

\author{Arun Kumar Das}
\email{akd.qip@gmail.com}
\affiliation{S. N. Bose National Centre for Basic Sciences, Block JD, Sector III, Salt Lake, Kolkata 700 106, India}

\author{Amit Kundu}
\email{amit8967@gmail.com}
\affiliation{S. N. Bose National Centre for Basic Sciences, Block JD, Sector III, Salt Lake, Kolkata 700 106, India}

\author{Pratik Ghosal}
\email{ghoshal.pratik00@gmail.com}
\affiliation{S. N. Bose National Centre for Basic Sciences, Block JD, Sector III, Salt Lake, Kolkata 700 106, India}
    
\begin{abstract}
Bipartite quantum states with higher Schmidt numbers have been shown to outperform those with lower Schmidt numbers in various quantum information processing tasks, highlighting the operational advantage of entanglement dimensionality. Certifying the Schmidt number of such states is therefore crucial for efficient resource utilisation. Ideally, this certification should rely as little as possible on the certifying devices to ensure robustness against their potential imperfections. Fully device-independent certification via Bell-nonlocal games offers strong robustness but suffers from fundamental limitations: it cannot certify the Schmidt number of all entangled states. We demonstrate that this insufficiency of Bell-nonlocal games is not limited to entangled states that do not exhibit Bell-nonlocality. Specifically, we prove the existence of Bell-nonlocal states whose Schmidt number cannot be certified by any Bell-nonlocal game when the parties are restricted to local projective measurements. To overcome this, we develop a measurement-device-independent certification method based on semiquantum nonlocal games, which assume trusted preparation devices but treat measurement devices as black boxes. We prove that for any bipartite state with Schmidt number exceeding $r$, there exists a semiquantum nonlocal game that can certify its Schmidt number. Finally, we provide an explicit construction of such a semiquantum nonlocal game based on an optimal Schmidt number witness operator.
\end{abstract}
\maketitle

\section{Introduction}\label{intro}

Entanglement is a distinctive feature of composite quantum systems that marks the fundamental distinction between quantum physics and classical theory \cite{horodecki2009quantum}. Beyond its profound foundational significance, entanglement is equally important---if not more so---from an operational perspective. It has been established as a key resource in various information processing tasks, including quantum cryptography \cite{PhysRevLett.67.661, PhysRevLett.83.648}, quantum teleportation \cite{PhysRevLett.70.1895}, entanglement-assisted quantum communication \cite{PhysRevLett.69.2881, bennett2002entanglement}, and quantum computation \cite{jozsa2003role}, among others, providing a ``quantum advantage'' over classical means.

Interestingly, the performance of such tasks depends not only on the presence of entanglement but also on its nature and degree. 
For instance, in quantum teleportation and dense-coding tasks, entangled states with zero distillable entanglement---called the bound entangled states---do not provide any advantage \cite{PhysRevA.60.1888, horodecki2001classical, 10315956}. On the other hand, in quantum cryptography, bound entangled states can be useful for distilling private key under one-way communication \cite{PhysRevLett.94.160502, 4529274}, whereas antidegradable (symmetric extendible) states are not \cite{PhysRevA.79.042329}. Therefore, for optimal performance in a given task, it is essential to choose entangled states with the appropriate characteristics.

In general, the full characterization of the entanglement of a given state is not straightforward; even in the simplest case of bipartite systems, there exist multiple independent and inequivalent entanglement measures, each quantifying a different aspect of it \cite{Plenio:2007zz}. One particular aspect of entanglement that we focus on in this paper is its dimensionality. The \textit{dimension of entanglement} refers to the number of degrees of freedom that actively contribute to the entanglement, rather than simply the local dimensions of the subsystems. For bipartite systems, this is quantified by the \textit{Schmidt number} of the state \cite{terhal2000schmidt}. High-dimensional entanglement is more robust against noise than low-dimensional entanglement and has been shown to offer advantages in a variety of tasks \cite{PhysRevA.71.044305, lanyon2009simplifying, PhysRevA.88.032309, Mirhosseini_2015, bae2019more, cozzolino2019high, wang2020qudits, kues2017chip, erhard2018twisted, erhard2020advances, cerf2002security, sheridan2010security, cozzolino2019orbital, bouchard2018experimental}. Its experimental realisation has also been achieved \cite{dada2011experimental, malik2016multi}.

Given the resourcefulness of high-dimensional entanglement, a natural question that arises is its certification---specifically, determining whether a given bipartite state $\rho_{AB}$ has Schmidt number greater than a certain value, say $r$. This problem is of considerable practical importance and has attracted significant attention in recent years. Several methods have been proposed to address this question, each exploiting different properties of quantum states \cite{bavaresco2018measurements, lib2024experimental, dkabrowski2018certification, huang2016high, liu2023characterizing, tavakoli2024enhanced, zhang2024analyzing, nzrc-8yrt}.

Perhaps the simplest and most commonly used approach is based on Schmidt number witnesses \cite{sanpera2001schmidt, chruscinski2014entanglement, bavaresco2018measurements, Shi_2024, li2025high}. States with Schmidt number less than or equal to $r$ form a convex and compact set \cite{terhal2000schmidt}. For any state lying outside this set, there exists a Hermitian operator called the Schmidt number witness, whose expectation value is strictly negative for that state, while for all states within the set, it is positive semidefinite. One can thus certify whether the Schmidt number of a given state exceeds $r$ by decomposing the corresponding witness operator into local observables, measuring them, and evaluating the expectation value of the witness from the observed data. However, a practical drawback of this method is that it requires perfect implementation of the measurement devices to accurately evaluate the expectation value. Inaccuracies in this implementation can lead to false positives:  states with a lower Schmidt number being falsely certified to be having a higher Schmidt number. One possible way to address this issue is to fully characterize all the devices involved in the certification procedure and to account for all possible sources of error that may arise, which is extremely challenging in practice.\footnote{ Since one needs to have a precise knowledge about the internal workings of the devices in order to do so.} A more practical alternative is to seek methods that rely as little as possible on the certifying devices.

In this paper, we explore such methods. A fully device-independent (DI) approach does not rely on the internal workings of any certifying devices---making it inherently robust to their imperfections---and instead certifies relevant properties of the system solely from the observed experimental probabilities. For Schmidt number certification---as for entanglement certification---the most natural setting to consider is that of a Bell-nonlocal game \cite{PhysicsPhysiqueFizika.1.195, RevModPhys.86.419}. In this framework, two spatially separated parties (verifiers) share the state whose Schmidt number is to be certified. They receive classical indices as inputs. Based on these inputs, each party performs local measurements on their respective subsystems, yielding classical outputs. From the resulting input-output correlation statistics, one constructs an inequality that is satisfied by all states with Schmidt number less than or equal to $r$. Violation of this inequality implies that the shared state has Schmidt number greater than $r$. For example, the set of inequalities used to witness \textit{Bell-nonlocality} \cite{RevModPhys.86.419} forms a specific subclass whose violation certifies that the shared state is entangled---that is, it has Schmidt number greater than one.

However, this approach generally cannot certify the Schmidt number of all entangled states (see also \cite{hirsch2020schmidt}). A simple argument for this limitation stems from the existence of entangled states that admit local hidden variable (LHV) models \cite{werner1989quantum, barrett2002nonsequential}. Such states, while entangled, do not violate any Bell inequality, rendering their entanglement undetectable within this framework. Since entanglement certification, i.e., verifying whether a state has Schmidt number greater than one, is a special case of Schmidt number certification, the limitations of the Bell-nonlocal-game-based approach naturally extend to the more general task. At this point, it is reasonable to question whether these limitations are specific to entangled states that do not exhibit Bell-nonlocality. In other words, for any state that \textit{is} Bell-nonlocal, is the Bell-nonlocal-game-based approach sufficient to certify that its Schmidt number is greater than $r$, for all values of $r \geq 2$?\footnote{Since these states are by definition nonlocal, Bell-nonlocal games can already certify that they have a Schmidt number greater than $1$.} In this work, we partially answer this question in the negative.

Nevertheless, we show that the Schmidt number of all entangled states can be certified by relaxing the demand for full device-independence to a partially device-independent method. We propose a measurement-device-independent (MDI) Schmidt number certification method using the framework of \textit{semiquantum nonlocal games} \cite{buscemi2012all}. Such games generalize the Bell-nonlocal game setting by replacing classical inputs with classically indexed quantum inputs. While trust on the devices preparing the input quantum states is assumed, this method remains independent of the specific measurement devices used by the parties. Trusting the preparation devices is a somewhat more natural assumption than trusting the measurement devices, as detection devices are generally open to the external environment, which can introduce imperfections due to environmental influence. In previous works,  the semiquantum nonlocal game framework has been adopted for MDI certification \cite{branciard2013measurement} and characterization of entanglement \cite{PhysRevA.94.012343, PhysRevLett.118.150505, goh2016measurement, PhysRevA.98.052332, PhysRevLett.132.110204, PhysRevA.95.042340}, as well as for MDI certification of beyond-quantum correlations \cite{PhysRevA.106.L040201, yu2024measurement}.\footnote{This framework has also been invoked in the context of steering \cite{PhysRevA.87.032306,kocsis2015experimental,ku2018measurement,PhysRevA.101.012333}, non-classical teleportation \cite{PhysRevLett.119.110501}, and classical simulation of quantum correlations \cite{rosset2013entangled, PhysRevA.88.032118}.} 

Apart from the fully DI and MDI approaches, there also exist certification methods within the semi-device-independent (semi-DI) framework, where the requirement of device independence is placed on only one party. Quantum steering \cite{PhysRevLett.98.140402} naturally fits into this setting. Entanglement can be certified by the violation of certain steering inequalities \cite{PhysRevA.87.062103,PhysRevA.90.050305,PhysRevA.93.012108,PhysRevA.95.032128,PhysRevA.93.012108}. More recently, Schmidt-number certification in steering scenarios has also been introduced in \cite{PhysRevLett.126.200404}, which proposed a two-setting steering inequality whose violation certifies genuinely high-dimensional steering unattainable with lower-dimensional entanglement. This was extended in \cite{PhysRevLett.131.010201}, where necessary and sufficient conditions for high-dimensional entanglement certification were provided via a hierarchy of semidefinite programs, though with rapidly increasing computational cost. To address this, \cite{PhysRevLett.134.090802} introduced a scheme whose computational cost is independent of the Schmidt number, thereby enabling efficient certification.

In this work, however, we restrict attention to MDI certification methods.

An outline of the paper and a summary of our main results are as follows:
\begin{itemize}
    \item In Sec .~\ref {SN}, we introduce the concept of the Schmidt number for bipartite states, along with Schmidt number witness operators, as a prerequisite.
    
    \item In Sec .~\ref {Bell-DI}, we describe the device-independent approach to Schmidt number certification via Bell-nonlocal games and analyze its limitations. We prove a general no-go theorem: there exist Bell-nonlocal states with Schmidt number greater than $r$ (where $r \geq 2$) that cannot be certified through any Bell-nonlocal game if the parties are restricted to local projective measurements [Theorem~\ref{theo1}]. As a concrete example, we present a family of Bell-nonlocal states with Schmidt number $3$ whose correlations under such measurements can be simulated by states with Schmidt number at most $2$. We extend this result to general positive-operator-valued measurements (POVMs) in a specific setting, namely, two-output Bell-nonlocal games [Corollary~\ref{cor1}].
    
    \item In Sec.~\ref{MDI}, we present our measurement-device-independent Schmidt number certification method. We prove the existence of a semiquantum nonlocal game for any bipartite state $\rho_{AB}$ with Schmidt number greater than $r$, which can be used to certify its Schmidt number [Theorem~\ref{theo2}]. We illustrate how to construct such semiquantum games from Schmidt number witnesses and explicitly discuss a concrete example using an optimal Schmidt number witness operator.

    \item Finally, we conclude in Sec.~\ref{conclusion} with a discussion of possible future directions.
\end{itemize}

\section{Schmidt number of bipartite states and its witness}\label{SN}

Consider a pure bipartite quantum state $\ket{\psi}_{AB}\in \mathbb{C}^{d_A} \otimes \mathbb{C}^{d_B}$. It can be expressed in Schmidt decomposition form \cite{peres1997quantum, nielsen2010quantum} as:
\begin{equation}
    \ket{\psi}_{AB}= \sum_{i=1}^k \sqrt{{\lambda}_i} {\ket{u_i}}_A {\ket{v_i}}_B, 
\end{equation}
where $\lambda_i \geq 0$ for all $i$, $\sum_i \lambda_i = 1$, and $\{\ket{u_i}_A\}$ and $\{\ket{v_i}_B\}$ are orthonormal bases for $\mathbb{C}^{d_A}$ and $\mathbb{C}^{d_B}$, respectively. The quantities $\sqrt{\lambda_i}$ are the Schmidt coefficients of $\ket{\psi}_{AB}$, and $1\leq k \leq \min \{d_A, d_B\}$ denotes the number of non-zero Schmidt coefficients, called the Schmidt rank of the state. The state $\ket{\psi}_{AB}$ is entangled if and only if its Schmidt rank is strictly greater than 1.

Apart from determining whether a state is entangled or not, the Schmidt rank also quantifies how many of the local dimensions effectively contribute to the entanglement. For example, consider the following $\mathbb{C}^3\otimes\mathbb{C}^3$ state
\begin{align*}
    \ket{\psi}_{AB}=&\frac{1}{2}\ket{0}_A\ket{0}_B+\frac{1}{4}\ket{0}_A\ket{1}_B+\frac{1}{4}\ket{0}_A\ket{2}_B+\frac{1}{2}\ket{1}_A\ket{0}_B\\
    &+\frac{1}{2}\ket{2}_A\ket{0}_B-\frac{1}{4}\ket{2}_A\ket{1}_B-\frac{1}{4}\ket{2}_A\ket{2}_B.
\end{align*}
One can verify from its Schmidt decomposition:
\begin{align*}
    \ket{\psi}_{AB}=\frac{\sqrt{3}}{2}\cdot\frac{1}{\sqrt{3}}\left(\ket{0}+\ket{1}+\ket{2}\right)_A\ket{0}_B+\frac{1}{2}\cdot\frac{1}{\sqrt{2}}\left(\ket{0}-\ket{2}\right)_A\frac{1}{\sqrt{2}}\left(\ket{1}+\ket{2}\right)_B,
\end{align*}
that its Schmidt rank is 2. Thus, although the local dimensions of the state are 3, its entanglement is effectively confined to a 2-dimensional subspace. In this sense, $\ket{\psi}_{AB}$ is no different in its entanglement content than the $\mathbb{C}^2 \otimes \mathbb{C}^2$ state $\ket{\phi}_{AB}=\frac{\sqrt{3}}{2}\ket{0}_A\ket{0}_B+\frac{1}{2}\ket{1}_A\ket{1}_B$. In other words, the Schmidt rank of a pure state captures its \textit{entanglement dimension}.

\begin{figure}[t!]
\centering
\includegraphics[scale=0.4]{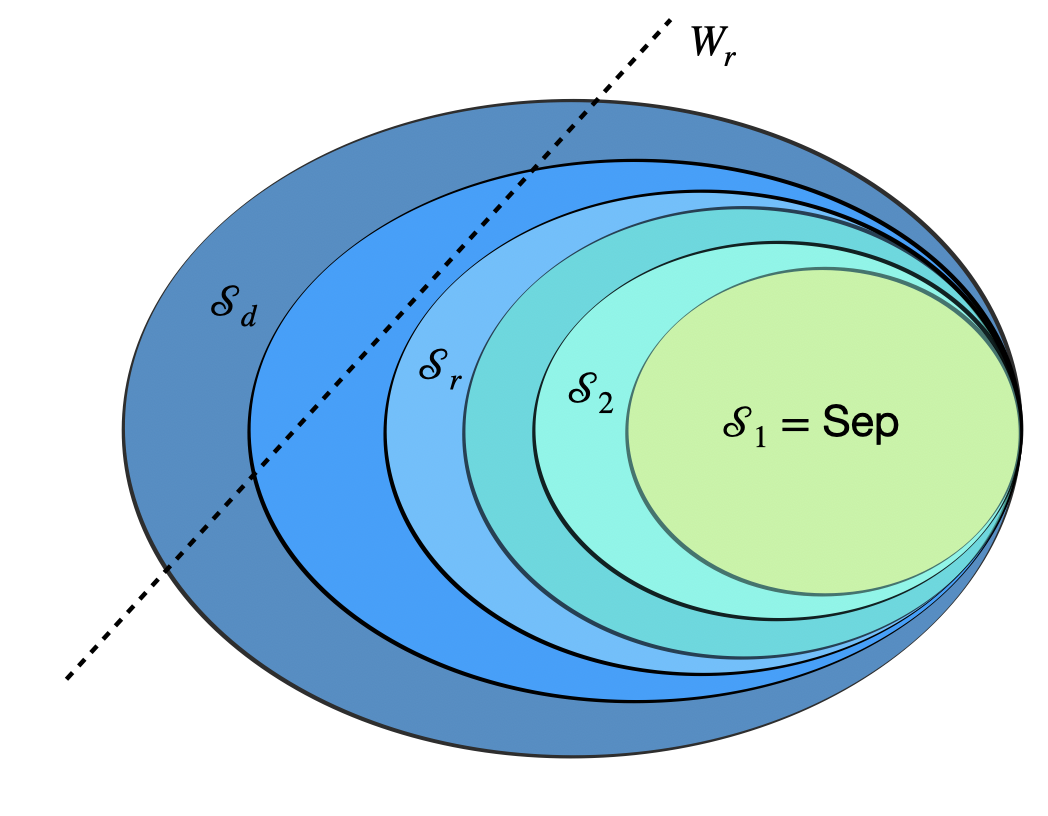}
\caption{Illustration of the nested subset structure of the sets $\mathcal{S}_r \subseteq \mathcal{D}(\mathbb{C}^{d_A}\otimes \mathbb{C}^{d_B})$ consisting of states with Schmidt number at most $r$, where $1 \leq r \leq d$ and $d=\min\{d_A,d_B\}$. The outermost set $\mathcal{S}_d$ corresponds to all states. $W_r$ is a witness operator that detects states with Schmidt number greater than $r$: its expectation value is positive semidefinite for all states lying right to it and strictly negative for any state lying left to it.}\label{fig1}
\vspace{-.5cm}
\end{figure}

Schmidt number generalizes this concept to mixed states \cite{terhal2000schmidt}. A mixed bipartite state $\rho_{AB}$ can be written (non-uniquely) as a convex combination of pure states:
$$\rho_{AB}= \sum_{k} p_k {\ket{{\psi_k}}}_{AB} {\bra{{\psi}_k}},$$ 
with $\ket{\psi_k}_{AB}\in\mathbb{C}^{d_A}\otimes \mathbb{C}^{d_B}$, $p_k \ge 0$ and $\sum_k p_k=1$. Its Schmidt number is defined as
\begin{equation}
    \text{SN}(\rho_{AB}) = \min_{\{p_{k}, \ket{\psi_{k}}_{AB}\}} \hspace{0.1cm}\left[\max_{k} \hspace{0.1cm} \text{SR}(\ket{{\psi_k}}_{AB})\right]
\end{equation}
where $\text{SR}(\ket{{\psi_k}}_{AB})$ denotes the Schmidt rank of $\ket{{\psi_k}}_{AB}$ and the minimization is over all possible pure state decompositions of $\rho_{AB}$. It follows that for any state $\rho_{AB}$, $1 \leq \text{SN}(\rho_{AB}) \leq \min \{d_A, d_B\}$, and for any separable state $\rho^{\text{sep}}_{AB}$, $\text{SN}(\rho^{\text{sep}}_{AB}) = 1$. Beyond characterizing the entanglement dimension of a system, the Schmidt number of a state is also closely connected to the entanglement of formation. In particular, the one-shot zero-error entanglement cost of a state is determined by the logarithm of its Schmidt number \cite{buscemi2011entanglement}.

The definition of Schmidt number extends naturally to positive operators, which are simply unnormalized density operators. As such, a positive operator has the same Schmidt number as the corresponding normalized state. This extension will be important for some of our results.

Let us now consider the set of density operators on $\mathbb{C}^{d_A}\otimes\mathbb{C}^{d_B}$ having Schmidt number $r$ or less, denoted by $\mathcal{S}_r$. This set is convex and compact within the real vector space of density operators, and satisfies the following chain of inclusions \cite{terhal2000schmidt}:
\begin{align}
    \mathcal{S}_1 \subset \mathcal{S}_2 \subset\cdots \subset \mathcal{S}_r\cdots\subset\mathcal{S}_d,~~\text{where}~d=\min\{d_A,d_B\}.  
\end{align}
Given a state $\rho_{AB}$, one can determine whether its Schmidt number exceeds $r$ by invoking the Hahn–Banach separation theorem \cite{holmes2012geometric}, which ensures that for any state $\rho_{AB}\notin\mathcal{S}_r$ there exists a Hermitian operator $W_r$, such that
\begin{align}\label{witness}
    & \mathrm{Tr}[W_r \rho_{AB}] < 0, \nonumber \\
    \text{and} \quad & \mathrm{Tr}[W_r \sigma_{AB}] \geq 0, \quad \forall \sigma_{AB} \in \mathcal{S}_r.
\end{align}
The operator $W_r$ is referred to as a \textit{Schmidt number witness} \cite{chruscinski2014entanglement, Shi_2024, sanpera2001schmidt}, as the real-valued linear functional $\mathrm{Tr}[W_r (\cdot)]$ defines a hyperplane that separates states outside $\mathcal{S}_r$ from those within (see Fig. \ref{fig1}).

Note that, since the Schmidt number witness is a linear functional, a single witness cannot, in general, detect all states outside $\mathcal{S}_r$. However, it remains an open problem how to construct an appropriate witness \(W_r\) tailored to a given state \(\rho_{AB}\). Even in the special case \(r = 1\), the existence of such a construction---i.e., generating an entanglement witness directly from the full classical description of \(\rho_{AB}\)---would essentially solve the problem of deciding whether \(\rho_{AB}\) is entangled. Yet, it is well established that detecting entanglement from the classical description of a state is NP-hard \cite{gharibian2008strong}. This complexity extends to higher values of \(r\) as well. Consequently, an explicit functional dependence of \(W_r\) on \(\rho_{AB}\) remains unknown. Nevertheless, Ref.~\cite{sanpera2001schmidt} introduced a family of Schmidt-number witnesses, called \textit{optimal Schmidt-number witnesses}, which are independent of \(\rho_{AB}\) but depend explicitly on \(r\). We discuss these witnesses in Sec.~\eqref{constructionoptimal}.

In practice, to verify whether a given bipartite state $\rho_{AB}$ has Schmidt number greater than $r$, two spatially separated parties can select an appropriate Schmidt number witness operator $W_r$ that detects the state. While the operator $W_r$ itself is not directly measurable, it can always be decomposed as a linear combination of locally measurable observables:
\begin{align}
W_r = \sum_i \alpha_i~(P^i_A \otimes Q^i_B),
\end{align}
where $\{P^i_{A}\}_i$ and $\{Q^i_{B}\}_i$ are observables on $\mathcal{H}_{A}$ and $\mathcal{H}_{B}$, respectively, and $\alpha_i$'s are real coefficients. By locally measuring the observables $P^i_A$ and $Q^i_B$, the parties can estimate the expectation values $\langle P^i_A \otimes Q^i_B\rangle_{\rho_{AB}}$ from experimental statistics. Using these, they can compute the expectation value of the witness $\langle W_r\rangle_{\rho_{AB}}=\sum_i \alpha_i~\langle P^i_A \otimes Q^i_B\rangle_{\rho_{AB}}$, and verify that $\rho_{AB}\notin\mathcal{S}_r$.

However, as mentioned earlier, this method of certification is not infallible against imperfections of the measurement devices. In the following section, we discuss a fully device-independent certification approach based on Bell-nonlocal games and examine its limitations.

\section{Device-Independent Schmidt Number Certification Using Bell-nonlocal games and its limitation}\label{Bell-DI}

Consider two distant parties (verifiers), Alice and Bob, who share a bipartite quantum state $\rho_{AB}$ and aim to verify, in a fully device-independent (DI) manner, whether the Schmidt number of $\rho_{AB}$ exceeds a given threshold $r$. Bell-nonlocal games provide a natural framework for such DI certification. In this setting, Alice and Bob are each provided with classical inputs $x$ and $y$, sampled from two independent finite index sets $\mathcal{X}$ and $\mathcal{Y}$, according to probability distributions $p(x)$ and $q(y)$, respectively. Upon receiving their inputs, they perform local measurements $\{M^{a|x}_A\}$ and $\{N^{b|y}_B\}$ on their respective subsystems, producing classical outputs $a$ and $b$ corresponding to the measurement outcomes (see Fig.~\ref{bell-nonlocal}). While the parties may share an arbitrary amount of pre-established classical randomness, they are not allowed to communicate during the game.

Repeating this game a sufficient number of times, the following correlation is generated:
\begin{align}
    p(a,b|x,y) = \Tr\left[(M^{a|x}_A \otimes N^{b|y}_B) \, \rho_{AB}\right].
\end{align}
From this correlation, Alice and Bob can conclude that $\text{SN}(\rho_{AB}) > r$ if no state $\sigma_{AB}$ with $\text{SN}(\sigma_{AB}) \leq r$ can reproduce the correlation, regardless of the choice of local measurements.

\begin{figure}[t!]
\centering
\includegraphics[scale=0.6]{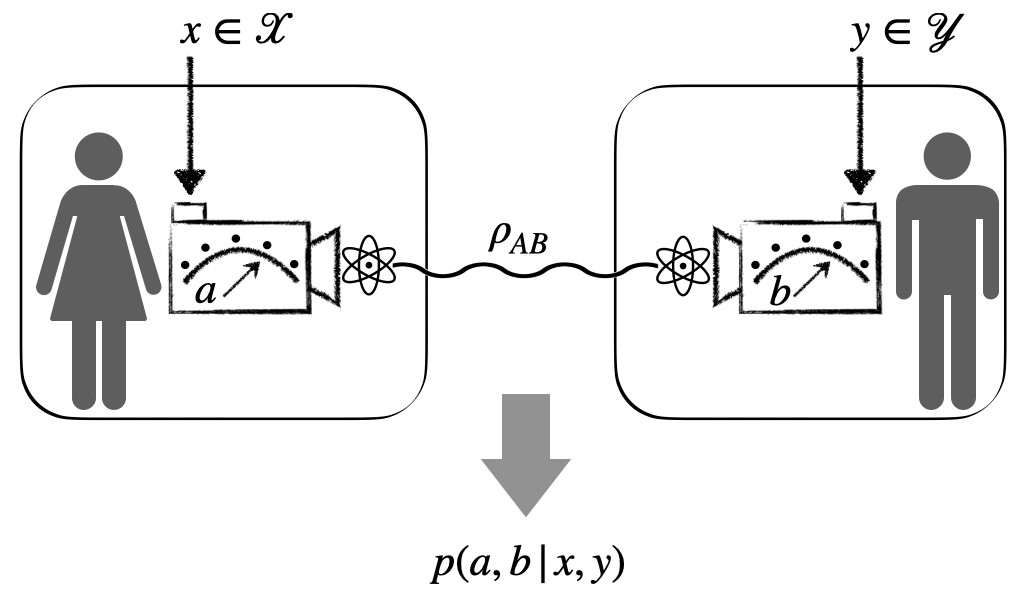}
\caption{A schematic setup for the certification of higher Schmidt number states in a standard Bell scenario with classical inputs and outputs.}\label{bell-nonlocal}
\vspace{-.5cm}
\end{figure}

To determine whether this condition holds, one can construct a linear inequality that must be satisfied by all correlations generated from states with Schmidt number $r$ or less, and which would be violated by the correlation obtained from $\rho_{AB}$ for a suitably chosen set of local measurements. The inequality is constructed as follows: In each run of the game, a \textit{payoff} function \( \mathscr{J}(a,b,x,y) \) is assigned based on the inputs \( (x,y) \) and outputs \( (a,b) \). The average payoff is then given by
\begin{align}
    \mathscr{J}_{\text{avg}} = \sum_{a,b,x,y} \mathscr{J}(a,b,x,y) \, p(a,b|x,y) \, p(x) \, q(y). \label{avgpayoff}
\end{align}
If $\mathscr{J}_{\text{avg}}$ is bounded for all correlations arising from states with $\text{SN} \leq r$, and the value obtained using $\rho_{AB}$ lies outside this bound, then this bound defines a linear inequality. For example, if the payoff function is chosen such that $\mathscr{J}_{\text{avg}} \geq 0$ for all states $\sigma_{AB}$ with $\text{SN}(\sigma_{AB}) \leq r$, then any observed violation of this inequality (i.e., $\mathscr{J}_{\text{avg}} < 0$) directly implies that $\text{SN}(\rho_{AB}) > r$.

Note that the conclusion about the Schmidt number follows solely from the input-output correlation statistics, without any assumptions about the internal workings of the certifying devices. Hence, the method is fully device-independent. Moreover, since the inequality cannot be violated by any state with Schmidt number $\leq r$, regardless of the choice of local measurements, this approach rules out the possibility of false positives due to device imperfections.  

A notable limitation of this method is that it cannot certify the Schmidt number of all entangled states. For instance, consider the two-qutrit isotropic states:
\begin{equation}
	\rho^{\text{iso}}_{AB}(\lambda) = \lambda |\Phi_3^{+}\rangle \langle \Phi_3^{+}| + (1-\lambda) \frac{\mathbb{I}_A}{3} \otimes \frac{\mathbb{I}_B}{3},\label{two-qutrit-iso}	
\end{equation}
where \( |\Phi_3^{+}\rangle = \frac{1}{\sqrt{3}} (\ket{00} + \ket{11} + \ket{22}) \) and \( \lambda \in [0,1] \). These states are entangled for \( \lambda > \frac{1}{4} \). However, for \( \lambda \le \frac{8}{27} \), the state admits a local hidden variable (LHV) model \cite{barrett2002nonsequential, augusiak2014local}, meaning that its correlations, for any local measurements, can be simulated by a separable state.

Consequently, in the range \( \frac{1}{4} < \lambda \le \frac{8}{27} \), although the states \( \rho^{\text{iso}}_{AB}(\lambda) \) are entangled, no Bell-nonlocal game can certify their entanglement. Since this framework cannot even establish that \( \text{SN}(\rho^{\text{iso}}_{AB}(\lambda)) > 1 \), it also fails to certify \( \text{SN}(\rho^{\text{iso}}_{AB}(\lambda)) > r \) for any \( r \ge 1 \). This limitation applies more broadly to all entangled states that do not exhibit Bell-nonlocality.

It is then natural to ask about the status of Bell-nonlocal states: are Bell-nonlocal games sufficient to certify the Schmidt number of all such states, for all values of \( r \)? In a restricted setting, we show that this is not the case.

\begin{theorem}\label{theo1}
    There exist quantum states that exhibit Bell-nonlocality, yet for which the fact that their Schmidt number exceeds $r$ (with $r \geq 2$) cannot be certified through any Bell-nonlocal game, if the parties are restricted to perform only local projective measurements.
\end{theorem}

\begin{proof}
    The proof is constructive. Consider the following family of states on $\mathbb{C}^8 \otimes \mathbb{C}^8$:
    \begin{align}
        \tilde{\rho}_{AB}(\lambda) = \lambda~ \rho^{\text{iso}}_{AB}(p=0.24) + (1 - \lambda) \ket{\phi^{+}_{2}}_{AB} \bra{\phi^{+}_{2}}, \quad \lambda \in [0,1], \label{ex2}
    \end{align}
    where
    \begin{align*}
        &\rho^{\text{iso}}_{AB}(p) = p~\ket{\Phi_8^{+}}_{AB} \bra{\Phi_8^{+}} + (1-p)\frac{\mathbb{I}_A}{8} \otimes \frac{\mathbb{I}_B}{8}, \quad p \in [0,1], \\
        &\ket{\Phi_8^+}_{AB} = \frac{1}{\sqrt{8}} \sum_{i=0}^7 \ket{ii}_{AB}, \quad
        \ket{\phi_2^+}_{AB} = \frac{1}{\sqrt{2}} \sum_{i=0}^1 \ket{ii}_{AB}.
    \end{align*}

    The states $\tilde{\rho}_{AB}(\lambda)$ have the following properties:

    \begin{lemma}\label{lemma1}
        The Schmidt number of $\tilde{\rho}_{AB}(\lambda)$ is $2$ for $\lambda = 0$, and $3$ for all $\lambda > 0$.
    \end{lemma}

    \begin{lemma}\label{lemma2}
        $\tilde{\rho}_{AB}(\lambda)$ violates the CHSH Bell inequality \cite{clauser1969proposed} under local projective measurements for $0 \le \lambda < 0.312$.
    \end{lemma}

    \noindent Proofs of Lemmas \eqref{lemma1} and \eqref{lemma2} are provided in Appendix~\ref{Appendix A}. Lemma~\eqref{lemma1} establishes that $\tilde{\rho}_{AB}(\lambda)$ is entangled for all $\lambda > 0$, and Lemma~\eqref{lemma2} shows that it exhibits Bell-nonlocality in the interval $\lambda \in [0, 0.312)$ via local projective measurements.

    We now focus on the interval $\lambda \in (0, 0.312)$, where $\tilde{\rho}_{AB}(\lambda)$ has Schmidt number $3$ and exhibits Bell-nonlocality. We will prove that no Bell-nonlocal game restricted to local projective measurements can certify that $\tilde{\rho}_{AB}(\lambda) \notin \mathcal{S}_2$ in this regime.

    For this, we use the following lemma (proof in Appendix~\ref{Appendix A}):

    \begin{lemma}\label{lemma3}
        The state $\rho^{\text{iso}}_{AB}(p=0.24)$ admits a local hidden variable model under local projective measurements.
    \end{lemma}

    Since $\rho^{\text{iso}}_{AB}(p=0.24)$ is local, any nonlocality exhibited by $\tilde{\rho}_{AB}(\lambda)$ under local projective measurements must originate from the $\ket{\phi_2^+}_{AB}$ component. Thus, the resulting correlations can be simulated using a state of Schmidt number $2$ and shared classical randomness.

    More concretely, let $p(a,b|x,y)$ denote the correlation obtained from performing local projective measurements on $\tilde{\rho}_{AB}(\lambda)$. This can be written as:
    \begin{align}\label{SN2correlation}
        p(a,b|x,y) = \lambda~ p^{L}(a,b|x,y) + (1 - \lambda)~ p^{\phi^+_2}(a,b|x,y),
    \end{align}
    where $p^L(a,b|x,y)$ is the local correlation arising from the state $\rho^{\text{iso}}_{AB}(p=0.24)$, and $p^{\phi^+_2}(a,b|x,y)$ is generated by measuring the entangled two-qubit state $\ket{\phi^+_2}_{AB}$.

    A simulation strategy of $p(a,b|x,y)$ using a state of Schmidt number-$2$ proceeds as follows: with probability $(1 - \lambda)$, the parties share $\ket{\phi^+_2}_{AB}$ and perform the relevant projective measurements to obtain $p^{\phi^+_2}(a,b|x,y)$; with probability $\lambda$, they use shared classical randomness---the LHV model of $\rho^{\text{iso}}_{AB}(p=0.24)$---to simulate $p^L(a,b|x,y)$. The resulting mixture reproduces the overall correlation in Eqn.~\eqref{SN2correlation}.

    Hence, any correlation arising from local projective measurements on $\tilde{\rho}_{AB}(\lambda)$ can be simulated by a state with Schmidt number at most $2$. Consequently, no Bell-nonlocal game limited to projective measurements can certify that $\text{SN}(\tilde{\rho}_{AB}(\lambda)) > 2$ for $\lambda \in (0, 0.312)$. This completes the proof.
\end{proof}

While Theorem~\ref{theo1} is stated only for local projective measurements, it admits an extension to general positive-operator-valued measurements (POVM) in a specific class of Bell-nonlocal games. In particular, since any two-outcome (dichotomic) POVM can be expressed as a convex combination of projective measurements \cite{PhysRevLett.97.050503}, we obtain the following corollary. 
\begin{corollary}\label{cor1}
There exist quantum states that are Bell-nonlocal, but for which the fact that their Schmidt number exceeds $r$ (with $r \geq 2$) cannot be certified in any two-output Bell-nonlocal game, even when the parties are allowed to perform arbitrary POVMs.
\end{corollary}
Here, by ``two-output Bell-nonlocal games" we mean scenarios in which each party implements dichotomic measurements only.

The limitations outlined above show that Bell-nonlocal games, while powerful, are fundamentally restricted in their ability to certify the Schmidt number of all entangled---or even Bell-nonlocal---states, particularly when only local projective measurements are allowed. To overcome these limitations, we now consider a scenario where the requirement of full device-independence is relaxed.

\section{Measurement-Device-Independent Schmidt Number Certification with Trusted Quantum Inputs}\label{MDI}

In this section, we demonstrate that the Schmidt number of any quantum state can be certified in an MDI manner using semiquantum nonlocal games \cite{buscemi2012all}. In this framework, instead of receiving classical inputs $x$ and $y$, Alice and Bob are provided with quantum states $\psi^x_{A_0}$ and $\phi^y_{B_0}$, respectively (see Fig. \ref{sq-nonlocal}). Notably, the standard Bell-nonlocal game is recovered as a special case of this framework when the sets $\{\psi^x_{A_0}\}_x$ and $\{\phi^y_{B_0}\}_y$ consist of mutually orthogonal states. 
In the general setting, while Alice and Bob know the ensembles from which their input states are drawn, they do not know which specific state is received in each round. Upon receiving their respective states, Alice and Bob each perform a joint measurement---$\{M^a_{AA_0}\}$ and $\{N^b_{BB_0}\}$, respectively---on their input state and their half of the shared state $\rho_{AB}$, yielding outputs $a$ and $b$. The lack of knowledge about the input states implies that their measurement strategies cannot depend on the individual inputs. The correlation generated by repeating this game multiple times is given by
\begin{align}
    p(a, b|\psi^{x}, \phi^{y})=\Tr[(M_{A A_{0}}^{a} \otimes N_{B B_{0}}^{b}) (\psi^{x}_{A_{0}} \otimes \rho_{AB} \otimes \phi^{y}_{B_{0}}) ]. \label{prob}
\end{align}
As in the Bell scenario, Alice and Bob can then infer whether $\text{SN}(\rho_{AB})>r$ by evaluating a suitably constructed average payoff [Eqn.~\eqref{avgpayoff}], where the correlation $p(a,b|x,y)$ is now replaced by $p(a,b|\psi^x,\phi^y)$.

It is important to emphasize that this certification method relies critically on the assumption that the input quantum states are trusted---that is, the method is valid provided Alice and Bob receive the intended states $\{\psi^x_{A_0}\}$ and $\{\phi^y_{B_0}\}$ with no imperfections or deviations. As a result, this approach is not fully device-independent. Nevertheless, it remains measurement-device-independent, as it does not rely on any assumptions about the internal workings or reliability of Alice and Bob’s measurement devices. Consequently, the method is robust against imperfections in the local measurements.

\begin{figure}[t!]
\centering
\includegraphics[scale=0.6]{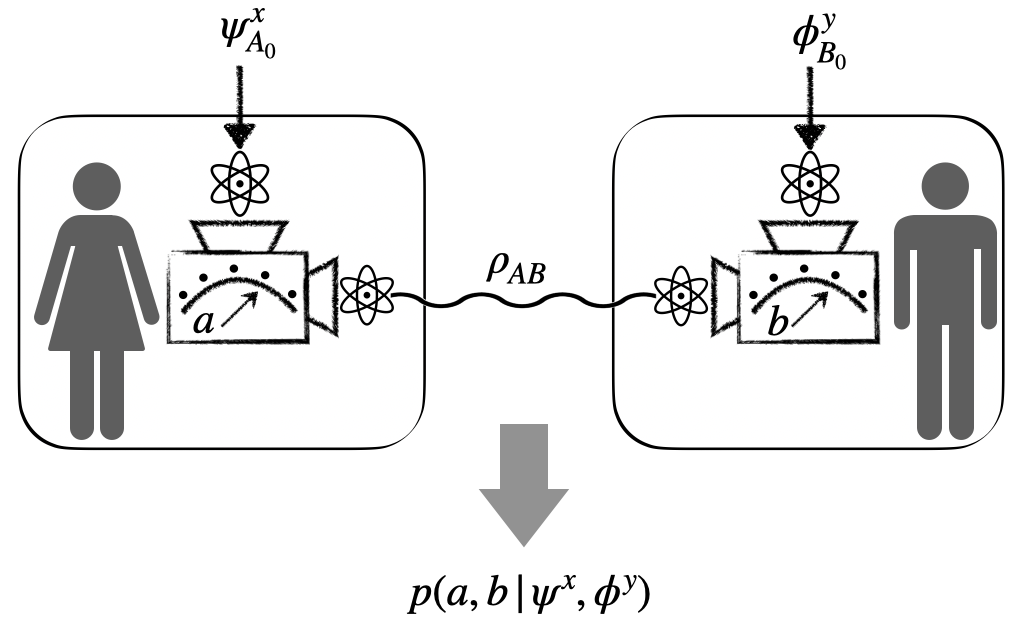}
\caption{A schematic setup for the certification of states with a higher Schmidt number in a measurement-device-independent framework using trusted quantum inputs and classical outputs. }\label{sq-nonlocal}
\vspace{-.5cm}
\end{figure}

We now prove that this MDI method can be used to certify the Schmidt number of all entangled states.
\begin{theorem}\label{theo2}
    For every bipartite quantum state $\rho_{AB}$ with $\text{SN}(\rho_{AB})>r$, there exists a semiquantum nonlocal game such that the average payoff for the game obtained using $\rho_{AB}$, $\mathscr{J}_{\text{avg}}(\rho_{AB})$ is strictly negative, whereas $\mathscr{J}_{\text{avg}}(\sigma_{AB}) \ge 0,~\forall\sigma_{AB} \in \mathcal{S}_{r}$.    
\end{theorem}

\begin{proof}
    We begin by noting that for every bipartite state $\rho_{AB}$ with $\text{SN}(\rho_{AB})>r$, there exists a Schmidt number witness operator $W_r$, such that Eqn. \eqref{witness} holds. Since $W_r$ is a Hermitian operator and the vector space of Hermitian operators can be spanned by the set of density operators, $W_r$ can be expressed (non-uniquely) as
    \begin{equation}
        W_{r} = \sum_{x, y} \gamma_{x,y}~(\xi^{x} \otimes \zeta^{y}), \label{witness-decomposed}
    \end{equation}
    where the coefficients $\gamma_{x,y}$ are real-valued for all $x, y$, and $\{\xi^{x}\}$ and $\{\zeta^{y}\}$ are density operators on $\mathcal{H}_A$ and $\mathcal{H}_B$, respectively. In general, these can be chosen to be informationally complete sets on the local Hilbert spaces.

    Based on this witness, we construct a semiquantum nonlocal game with the following specifications:
    \begin{itemize}
        \item The input states are given by $\psi_{A_0}^x = (\xi_{A_0}^x)^\intercal$ and $\phi_{B_0}^y = (\zeta_{B_0}^y)^\intercal$, with $\mathcal{H}_{A_0(B_0)} \cong \mathcal{H}_{A(B)}$;
        \item The outputs $a$ and $b$ each take values in $\{1, 2\}$;
        \item The payoff function $\mathscr{J}(a, b, x, y)$ and the probability distributions $p(x)$ and $q(y)$ satisfy
    \begin{equation}
    p(x)~q(y)~\mathscr{J}(a,b,x,y) = \gamma_{a,b,x,y} = \left\{
        \begin{array}{ll}
            \gamma_{x,y} & \text{if } (a,b) = (1,1), \\
            0 & \text{otherwise}.
        \end{array}
    \right.  \label{weightage}
    \end{equation}
    \end{itemize}
    The average payoff of this game is then given by
    \begin{align}
        \mathscr{J}_{\text{avg}}(\rho_{AB})=\sum_{x,y} \gamma_{x,y} \Tr[(M_{A_{0}A}^{a=1} \otimes N_{B_{0}B}^{b=1})((\xi_{A_{0}}^{x})^\intercal\otimes \rho_{AB} \otimes (\zeta_{B_{0}}^{y})^\intercal)]. \label{payoff3}
    \end{align}

    This game certifies that $\text{SN}(\rho_{AB}) > r$ if the following conditions are satisfied:
    \begin{enumerate}[$(i)$]
        \item For all measurement settings chosen by Alice and Bob, \(\mathscr{J}_{\text{avg}} (\rho_{AB}) \ge 0\) whenever \(\rho_{AB} \in \mathcal{S}_r\);
        \item If \(\rho_{AB} \notin \mathcal{S}_r\), then there exists at least one measurement setting for which \(\mathscr{J}_{\text{avg}} (\rho_{AB}) < 0\).\\
    \end{enumerate}

\textbf{Proof of the first condition:} 
Let $\rho_{AB} \in \mathcal{S}_r$. Then there exists a decomposition $\rho_{AB} = \sum_{k} p_k \ket{\eta_k}_{AB} \bra{\eta_k}$ such that $\text{SR}(\ket{\eta_k}) \le r$ for all $k$. Substituting this into Eqn.~\eqref{payoff3}, we get 
\begin{align}\label{pay}
    \mathscr{J}_{\text{avg}} (\rho_{AB})&= \sum_k p_k\sum_{x,y}\gamma_{x,y}\Tr[(M_{A_{0}A}^{a=1} \otimes N_{B_{0}B}^{b=1})\cdot((\xi_{A_{0}}^{x})^\intercal\otimes \ket{\eta_k}_{AB}\bra{\eta_k} \otimes (\zeta_{B_{0}}^{y})^\intercal)]\nonumber\\
    &=\sum_{k}p_{k}\sum_{x,y} \gamma_{x,y} \Tr_{A_{0}B_{0}} [R_{A_{0}B_{0}}^{(k)}((\xi_{A_{0}}^{x})^\intercal \otimes (\zeta_{B_{0}}^{y})^\intercal)],
\end{align} 
where
\begin{align}
    R_{A_0 B_0}^{(k)}= \Tr_{AB}[{(M_{A_{0}A}^{a=1}\otimes N_{B_{0}B}^{b=1})(\ket{\eta_{k}}_{AB}\bra{\eta_{k}}\otimes \mathbb{I}_{A_{0}B_{0}}})].
\end{align}

Note that $R_{A_0 B_0}^{(k)}$ is a positive operator proportional to the state obtained by performing the local measurements $\{M_{A A_0}^a\}$ and $\{N_{B B_0}^b\}$ on the joint state $\left(\ket{\eta_k}_{AB} \bra{\eta_k} \otimes \frac{\mathbb{I}_{A_0 B_0}}{d_{A_0} d_{B_0}}\right)$, and post-selecting on outcome $(a, b) = (1, 1)$. Since $\text{SR}(\ket{\eta_k}) \le r$, we have $\text{SN}(\ket{\eta_k}\bra{\eta_k} \otimes \mathbb{I}) \le r$ \cite{terhal2000schmidt}. As Schmidt number is monotonic under local filtering (a subset of SLOCC operations) \cite{ishizaka2004bound,10.5555/2011706.2011707,dur2000three,eltschka2014quantifying,schmid2008discriminating}, it follows that $\text{SN}(R_{A_0 B_0}^{(k)}) \le r$ for all $k$. By convexity of $\mathcal{S}_r$ \cite{terhal2000schmidt}, it follows that $\sum_k p_k R_{A_0 B_0}^{(k)} \in \mathcal{S}_r$.

Using Eqn.~\eqref{witness-decomposed} in Eqn.~\eqref{pay}, the average payoff becomes
\begin{align}
    \mathscr{J}_{\text{avg}} (\rho_{AB}) &= \Tr_{A_0 B_0}\left[\sum_k p_k R_{A_0 B_0}^{(k)} W_r^\intercal \right] \nonumber \\
    &= \Tr_{A_0 B_0} \left[\sum_k p_k (R_{A_0 B_0}^{(k)})^\intercal W_r\right] \ge 0. \label{finalpayoff}
\end{align}
In the last step, we used the fact that Schmidt number does not increase under transposition, i.e., $\sum_k p_k (R_{A_0 B_0}^{(k)})^\intercal \in \mathcal{S}_r$.

\textbf{Proof of the second condition:} 
Now consider the case $\rho_{AB} \notin \mathcal{S}_r$, i.e., $\text{SN}(\rho_{AB}) > r$. Let Alice and Bob's local measurements be
\[
\{M_{A A_0}^{a=1} = \mathcal{P}_{A A_0},\; M_{A A_0}^{a=2} = \mathbb{I} - \mathcal{P}_{A A_0}\} \quad \text{and} \quad \{N_{B B_0}^{b=1} = \mathcal{P}_{B B_0},\; N_{B B_0}^{b=2} = \mathbb{I} - \mathcal{P}_{B B_0}\},
\]
respectively, where $\mathcal{P}_{X X_0}$ denotes the projector onto the maximally entangled state $\ket{\Phi^+}_{X X_0} = \frac{1}{\sqrt{d_X}} \sum_{i=0}^{d_X - 1} \ket{ii}$, for $X = A, B$.

For this choice of measurements, the average payoff becomes
\begin{align}
    \mathscr{J}_{\text{avg}}(\rho_{AB}) &= \sum_{x,y} \gamma_{x,y} \Tr[(\mathcal{P}_{A A_0} \otimes \mathcal{P}_{B B_0}) ((\xi_{A_0}^x)^\intercal \otimes \rho_{AB} \otimes (\zeta_{B_0}^y)^\intercal)] \nonumber \\
    &= \frac{1}{d_A d_B} \sum_{x,y} \gamma_{x,y} \Tr_{AB}[(\xi^x_A \otimes \zeta^y_B) \rho_{AB}] \nonumber \\
    &= \frac{1}{d_A d_B} \Tr_{AB}[W_r \rho_{AB}] < 0.
\end{align}
Here, we used the identity $\Tr_{X_0}[\mathcal{P}_{X X_0} (\mathbb{I}_X \otimes C_{X_0}^\intercal)] = \frac{1}{d_X}C_X$.

Therefore, for every state $\rho_{AB}$ with $\text{SN}(\rho_{AB}) > r$, we can construct a semiquantum nonlocal game that certifies its Schmidt number for all values of $r$. This completes the proof.
\end{proof}

In what follows, we present an explicit construction of a semiquantum nonlocal game based on a standard Schmidt number witness operator.

\subsection{Construction of semiquantum nonlocal game from an Optimal Schmidt Number Witness}\label{constructionoptimal}

In general, a single witness operator \( W_r \) cannot detect all states outside the set \( \mathcal{S}_r \) using only one copy of the state---that is, there exists no \emph{universal} Schmidt number witness. However, one can define the notion of an \emph{optimal} Schmidt number witness operator \cite{sanpera2001schmidt}. Given two $r$-Schmidt number witness operators \( W_r \) and \( \tilde{W}_r \), we say that \( W_r \) is \emph{finer} than \( \tilde{W}_r \) if it detects a strictly larger set of states. A witness \( W_r \) is called \emph{optimal} if there exists no other witness that is finer than \( W_r \).

An example of an optimal $r$-Schmidt number witness on \( \mathbb{C}^d \otimes \mathbb{C}^d \) is 
\begin{align}
    W^{\text{opt}}_r = \mathbb{I}_d \otimes \mathbb{I}_d - \frac{d}{r} \mathcal{P}_d,
\end{align}
where \( \mathcal{P}_d \) is the projector onto the $d$-dimensional maximally entangled state $\ket{\Phi_d^+} = \frac{1}{\sqrt{d}} \sum_{i=0}^{d-1} \ket{ii}$.

We now provide an explicit construction of the semiquantum nonlocal game corresponding to this optimal witness in the case \( d = 3 \), following the method outlined in the proof of Theorem~\ref{theo2}. Note that for this dimension, the Schmidt number witness with \( r=2 \) represents the only nontrivial witness beyond the standard entanglement witness. The corresponding witness takes the form
\begin{equation}
    W^{\text{opt}}_2 = \mathbb{I}_3 \otimes \mathbb{I}_3 - \frac{3}{2} \mathcal{P}_3. \label{opt3}
\end{equation}

Using this witness operator, we can detect, for example, states with Schmidt number exceeding $2$ from the family of two-qutrit isotropic states in Eqn.~\eqref{two-qutrit-iso}. It can be verified that \( \rho^{\text{iso}}_{AB}(\lambda) \) has Schmidt number greater than $2$ for \( \lambda \in \left( \frac{5}{8}, 1 \right] \).

The witness $W^{\text{opt}}_2$ can be decomposed as $$W^{\text{opt}}_2 = \sum_{x,y=1}^9\gamma_{x, y}~(\zeta^{x}\otimes\zeta^{y}),$$ 
where the coefficients \(\gamma_{x,y}\) are represented in the form of a symmetric matrix given by:
\begin{align}
    \{\gamma_{x,y}\}_{x,y=1}^9 \equiv\begin{pmatrix}
        \frac{1}{2} & 1 & 1 & \frac{1}{4} & \frac{1}{4} & 0 & -\frac{1}{4} & -\frac{1}{4} & 0 \\
        1 & \frac{1}{2} & 1 & \frac{1}{4} & 0 & \frac{1}{4} & -\frac{1}{4} & 0 & -\frac{1}{4}\\
        1 & 1 & \frac{1}{2} & 0 & \frac{1}{4} & \frac{1}{4} & 0 & -\frac{1}{4} & -\frac{1}{4}\\
        \frac{1}{4} & \frac{1}{4} & 0 & -\frac{1}{4} & 0 & 0 & 0 & 0 & 0\\
        \frac{1}{4} & 0 & \frac{1}{4} & 0 & -\frac{1}{4} & 0 & 0 & 0 & 0\\
        0 & \frac{1}{4} & \frac{1}{4} & 0 & 0 & -\frac{1}{4} & 0 & 0 & 0\\
        -\frac{1}{4} & -\frac{1}{4} & 0 & 0 & 0 & 0 & \frac{1}{4} & 0 & 0\\
        -\frac{1}{4} & 0 & -\frac{1}{4} & 0 & 0 & 0 & 0 & \frac{1}{4} & 0\\
        0 & -\frac{1}{4} & -\frac{1}{4} & 0 & 0 & 0 & 0 & 0 & \frac{1}{4}
    \end{pmatrix}
\end{align}

and $\zeta^{x(y)}=\ketbra{\zeta^{x(y)}}{\zeta^{x(y)}}$'s are given by
\begin{align*}
    \begin{array}{ll}
            \ket{\zeta^1}= \ket{0}, \quad & \quad \ket{\zeta^{6}} = \frac{1}{\sqrt{2}}(\ket{1}+\ket{2}), \\
            \ket{\zeta^{2}} = \ket{1}, \quad & \quad \ket{\zeta^{7}} = \frac{1}{\sqrt{2}}(\ket{0} + i\ket{1}), \\
            \ket{\zeta^{3}} = \ket{2}, \quad & \quad \ket{\zeta^{8}} =  \frac{1}{\sqrt{2}}(\ket{0} + i\ket{2}), \\
            \ket{\zeta^{4}} = \frac{1}{\sqrt{2}}(\ket{0}+\ket{1}), \quad & \quad \ket{\zeta^{9}} =  \frac{1}{\sqrt{2}}(\ket{1} + i\ket{2}). \\
            \ket{\zeta^{5}} = \frac{1}{\sqrt{2}}(\ket{0}+\ket{2}),
        \end{array}
\end{align*}

Following the construction in Theorem~\ref{theo2}, Alice and Bob can certify that the two-qutrit state $\rho^{\text{iso}}_{AB}(\lambda) \notin \mathcal{S}_2$ for $\lambda \in (\frac{5}{8}, 1]$ using a two-output semiquantum nonlocal game with
\begin{itemize}
    \item uniformly chosen quantum inputs $\psi^x_{A_0} = (\zeta_{A_0}^x)^\intercal$ and $\phi^y_{B_0} = (\zeta_{B_0}^y)^\intercal$, and
    \item a payoff function $\mathscr{J}(a,b,x,y) = \left\{
        \begin{array}{ll}
            81\gamma_{x,y} & \text{if } a=b=1, \\
            0 & \text{otherwise}.
        \end{array}
    \right.$
\end{itemize}

\section{Conclusion}\label{conclusion}

In this work, we have addressed the problem of certifying whether the Schmidt number of a bipartite quantum state exceeds a given threshold $r$, focusing on both fully and partially device-independent methods. We first examined the fundamental limitations of fully device-independent certification using Bell-nonlocal games. While this approach offers strong robustness against device imperfections, we showed that it cannot certify the Schmidt number of all entangled states. Interestingly, this limitation arises not only from the existence of entangled states that do not exhibit Bell-nonlocality, but also extends to certain states which are nonlocal yet whose Schmidt number cannot be certified by any Bell-nonlocal game when the parties are restricted to local projective measurements. That is, the correlations generated by these states by local projective measurements in Bell-nonlocal games admit a low-entanglement-dimension simulation model. Furthermore, this can be extended to general POVMs in any two-output Bell-nonlocal games.

However, to fully understand the scope of this limitation, it is important to investigate whether our no-go result persists for general POVMs beyond two-output nonlocal settings. We leave this as an open question for future work. Should a positive answer be obtained---i.e., if the result holds under general POVMs---it would be interesting to explore whether, for all $2\leq r\leq d$, there always exists a state such that Bell-nonlocal games can certify that its Schmidt number exceeds $(r-1)$ but cannot certify that it exceeds $r$.

Next, we have presented a measurement-device-independent (MDI) method for Schmidt number certification using the framework of semiquantum nonlocal games, where the measurement devices are untrusted but the state preparation devices are assumed to be trusted. We proved that for any bipartite state on $\mathbb{C}^d\otimes \mathbb{C}^d$, there exists a semiquantum nonlocal game that can certify that the Schmidt number of the state exceeds $r$, for all $1\leq r\leq d$. Furthermore, as an illustration, we provided an explicit construction of such games based on optimal Schmidt number witness operators.

Our results demonstrate that not only the presence of entanglement, but also its dimensionality, can be certified in an MDI framework using semiquantum nonlocal games. This highlights the power of semiquantum nonlocal games to reveal fine-grained entanglement structure in scenarios where fully device-independent methods fall short. In a separate forthcoming work, we extend these results by developing an alternative device-independent approach for Schmidt number certification. However, it remains an open question whether similar certification methods can be devised to capture other types of entanglement hierarchies defined by different entanglement measures.

Finally, motivated by previous experimental demonstrations of MDI entanglement certification \cite{verbanis2016resource, kocsis2015experimental, xu2014implementation}, we emphasize that our proposed method is well-suited for experimental realization using current technologies [see also \cite{Chang2021}]. This opens the door to practical and robust certification of high-dimensional entanglement.

\begin{acknowledgements}
    We thank Sahil Gopalkrishna Naik, Ram Krishna Patra, Ananda G. Maity, Nirman Ganguly, and Manik Banik for helpful discussions and suggestions. We thank Subhendu B. Ghosh for pointing out Corollary \ref{cor1} to us. B.M. acknowledges the DST INSPIRE fellowship program for financial support.
\end{acknowledgements}

\bibliography{main}

\newpage
\appendix

\section{Proof of Lemmas}\label{Appendix A}

Consider the isotropic class of states $\rho^{\text{iso}}_{AB}(p)$ on $\mathbb{C}^d\otimes\mathbb{C}^d$ 
\begin{equation}
    \rho^{\text{iso}}_{AB}(p)= p |\Phi_d^{+}\rangle_{AB} \langle \Phi_d^{+}| + (1-p) \frac{\mathbb{I}_A}{d} \otimes \frac{\mathbb{I}_B}{d} \label{isotropicstate}
\end{equation}
where $\ket{\Phi_d^{+}}_{AB}=\frac{1}{\sqrt{d}}\sum_{i=0}^{d-1} \ket{ii}_{AB}$ and $p\in[0,1]$. This class of states has the following properties:
\begin{enumerate}
    \item \label{SNisotropic} SN$(\rho^{\text{iso}}_{AB}(p))\leq r$ if and only if $0 \le p \le \frac{rd-1}{d^2-1}$ \cite{terhal2000schmidt, mallick2024characterization}. This implies that $\rho^{\text{iso}}_{AB}(p)$ is entangled for $\frac{1}{d+1} < p \le 1$.
    \item \label{localmodel} $\rho^{\text{iso}}_{AB}(p)$ admits a LHV model for $p \le \frac{\sum_{k=2}^{d} \frac{1}{k}}{d-1}$ under local projective measurements \cite{augusiak2014local}.
\end{enumerate}

\noindent \textbf{Proof of Lemma \eqref{lemma3}:} For the state $\rho^{\text{iso}}_{AB}(p)$ on $\mathbb{C}^8\otimes\mathbb{C}^8$, using property \eqref{localmodel}, we get that the state has an LHV model under local projective measurements for $p\leq 0.245$. Therefore, $\rho^{\text{iso}}_{AB}(p=0.24)$ admits an LHV model under local projective measurements. \qed

Furthermore, from properties \eqref{SNisotropic} and \eqref{localmodel}, it follows that for $0.238 < p \le 0.245$, $\rho^{\text{iso}}_{AB}(p)$ on $\mathbb{C}^8\otimes\mathbb{C}^8$ has a Schmidt number $3$ as well as admits an LHV model under local projective measurements.\\

Now, consider the family of states on \( \mathbb{C}^8 \otimes \mathbb{C}^8 \) given by Eqn. \eqref{ex2}.

\noindent \textbf{Proof of Lemma \eqref{lemma1}:} Since \(\text{SR}(\ket{\phi^{+}_{2}}_{AB}) = 2\) and \(\text{SN}(\rho^{\text{iso}}_{AB}(p=0.24)) = 3\), it follows from the convexity of the set \(\mathcal{S}_3\) that \(\tilde{\rho}_{AB}(\lambda) \in \mathcal{S}_3\) \cite{terhal2000schmidt}. Now, consider the optimal Schmidt number witness operator on $\mathbb{C}^8\otimes\mathbb{C}^8$ \cite{sanpera2001schmidt}
\begin{equation} 
W^{\text{opt}}_r = \mathbb{I}_8 \otimes \mathbb{I}_8 - \frac{8}{r} \hspace{0.1cm} \mathcal{P}_d, \label{optimalwitness}
\end{equation}
where \(\mathcal{P}_d\) is the projector onto the maximally entangled state \(\ket{{\Phi_d}^+} = \frac{1}{\sqrt{8}} \sum_{i=0}^7 \ket{ii}\). The operator \(W^{\text{opt}}_r\) serves as a witness for whether a state has a Schmidt number greater than \(r\). For \(\tilde{\rho}_{AB} (\lambda)\), we get \(\text{Tr}[W^{\text{opt}}_2~\tilde{\rho}_{AB}] < 0\) for \(\lambda \in (0,1]\), which conclusively implies that \(\tilde{\rho}_{AB} \notin \mathcal{S}_2\) in this range. Therefore, we conclude that \(\text{SN}(\tilde{\rho}_{AB}) = 3\) for \(\lambda \in (0,1]\). 

For $\lambda=0$, $\Tilde{\rho}_{AB}(\lambda=0)=\ketbra{\phi_2^+}{\phi_2^+}$. Therefore, $\text{SN}(\Tilde{\rho}_{AB}(\lambda=0))=2$. \qed \\

\noindent \textbf{Proof of Lemma \eqref{lemma2}:} Let Alice and Bob share the bipartite state $\tilde\rho_{AB}(\lambda)$ on $\mathbb{C}^8\otimes\mathbb{C}^8$. The goal is to find the range of $\lambda$ such that the state exhibits nonlocal statistics under suitable choices of incompatible measurements on Alice's and Bob's sides. Let the inputs to Alice and Bob be represented by $x$ and $y$ respectively, where $x, y \in \{0,1\}$. The respective outcomes for Alice and Bob are denoted by $a$ and $b$, where  $a, b \in \{0,1,2\}$. For $x=0$, Alice performs the three-outcome projective measurement $$M_{0}= \{\ket{0}\bra{0}, \ket{1}\bra{1}, \mathbb{1}-\ket{0}\bra{0}- \ket{1}\bra{1}\}.$$ For $x=1$, Alice performs the projective measurement $$M_{1}= \{\ket{+}\bra{+}, \ket{-}\bra{-}, \mathbb{1}-\ket{+}\bra{+}- \ket{-}\bra{-}\}$$ where $\ket{\pm}= \frac{1}{\sqrt{2}}(\ket{0}\pm\ket{1})$. The corresponding measurements for Bob are $N_{0}=\{\ket{\alpha_{1}}\bra{\alpha_{1}}, \ket{\alpha_{2}}\bra{\alpha_{2}}, \mathbb{1}-\ket{\alpha_{1}}\bra{\alpha_{1}}-\ket{\alpha_{2}}\bra{\alpha_{2}}\}$
    where $\ket{\alpha_{1}}= \frac{1}{\sqrt{4+2\sqrt{2}}}\left((1+\sqrt{2})\ket{0}+\ket{1}\right)$, $\ket{\alpha_{2}}= \frac{1}{\sqrt{4-2\sqrt{2}}}\left((1-\sqrt{2})\ket{0}+\ket{1}\right)$, and $N_{1}= \{\ket{\beta_{1}}\bra{\beta_{1}}, \ket{\beta_{2}}\bra{\beta_{2}}, \mathbb{1}-\ket{\beta_{1}}\bra{\beta_{1}}-\ket{\beta_{2}}\bra{\beta_{2}\}}$
    where $\ket{\beta_{1}}= \frac{1}{\sqrt{4+2\sqrt{2}}}\left((1+\sqrt{2})\ket{0}-\ket{1}\right)$, $\ket{\beta_{2}}= \frac{1}{\sqrt{4-2\sqrt{2}}}\left((1-\sqrt{2})\ket{0}-\ket{1}\right)$. It is known that corresponding to every Bell inequality, there exists a nonlocal game and vice versa \cite{silman2008relation}. Let Alice and Bob play a game under these measurement settings and the state $\tilde\rho_{AB}(\lambda)$, with a payoff given by 
\begin{equation}
   \mathscr{J}(\rho) =    \left\{
\begin{array}{ll}
      +1 & \text{for} \hspace{2mm} a \oplus b = xy \hspace{2mm}\text{and} \hspace{2mm} a,b \in \{0,1\} \\
      -1 & \text{for} \hspace{2mm} a \oplus b = \overline{xy} \hspace{2mm}\text{and} \hspace{2mm} a,b \in \{0,1\} \\
      0 &  \text{for} \hspace{2mm} a,b \notin \{0,1\}
\end{array} 
\right.  \label{Bellpayoff}
\end{equation} 
It can be clearly seen that the average payoff of this game reduces to the Clauser-Horne-Shimony-Holt (CHSH) inequality \cite{clauser1969proposed}. This inequality surpasses the local bound of $2$ for $\lambda \in [0, 0.312)$ indicating that the state $\tilde\rho_{AB}(\lambda)$ is CHSH-nonlocal for $\lambda \in [0, 0.312)$. \qed

\end{document}